\documentclass[12pt,english]{article}

\usepackage[T1]{fontenc}

\usepackage{babel}
\usepackage{amsmath}
\usepackage{amsthm}

\makeatletter

\usepackage{paralist} 
\usepackage{pxfonts}
\usepackage{bbm}

\usepackage{dsfont}

\usepackage[onehalfspacing]{setspace}

\usepackage[a4paper, margin=1in]{geometry}

\usepackage{amsmath,amsfonts,amsthm, amssymb}

\usepackage{graphicx}
\usepackage{epstopdf}
\usepackage{tikz}

\usepackage{subfig}
\usepackage{mathtools}
\usepackage{booktabs}

\usepackage{siunitx}

\usepackage{url}

\usepackage{caption}

\usepackage[authoryear,round, sort]{natbib} 

\usepackage{hyperref}
\hypersetup{colorlinks=false,pdfborder=0 0 0}

\usepackage{cleveref}
\crefname{equation}{equation}{equations}
\crefname{figure}{figure}{figures}	
\crefname{table}{table}{tables}

\usepackage{enumitem}

\definecolor{winered}{RGB}{159,53,58}
\definecolor{RoyalBlue}{RGB}{159,53,58}

\usepackage{amsthm,thmtools,xcolor}

\newtheoremstyle{dotlessP}{}{}{\color{black}\itshape}{}{\color{black}\scshape}{}{ }{}

\theoremstyle{dotlessP}
\newtheorem{theorem}{Theorem}
\theoremstyle{dotlessP}

\theoremstyle{dotlessP}

\theoremstyle{dotlessP}
\theoremstyle{definition}
\theoremstyle{dotlessP}
\newtheorem{definition}{Definition}

\theoremstyle{dotlessP}

\theoremstyle{dotlessP}
\newtheorem{proposition}{Proposition}
\theoremstyle{dotlessP}

\theoremstyle{dotlessP}

\numberwithin{equation}{section}

\numberwithin{corollary}{section}

\hypersetup{
    bookmarks=true,         
    pdfstartview={FitH},    
    pdftitle={My title},    
       colorlinks=true,       
    linkcolor=winered,          
    citecolor=RoyalBlue,        
    filecolor=RoyalBlue,      
    urlcolor=RoyalBlue           
}
\newtagform{red}{\color{winered}(}{)}
\usetagform{red}

\usepackage[font={color=winered,footnotesize},labelsep=quad,width=.75\textwidth,labelfont={color=RoyalBlue}]{caption}
\usepackage{sectsty}
\sectionfont{\color{black}}  
\subsectionfont{\color{black}}  
\subsubsectionfont{\color{black}}  
\paragraphfont{\color{black}}  
\date{\small \today}

\newcommand*{\myproofname}{}

\newenvironment{myproof1}[1][\textit{\textcolor{winered}{Proof}}]{\begin{proof}[#1]}{\end{proof}}

\usepackage{changepage}

\usepackage{lscape}
\usepackage{multirow}
\usepackage{refstyle}
\newref{prop}{refcmd={\hyperref[#1]{Proposition \ref{#1}}}}
\newref{fig}{refcmd={\hyperref[#1]{Figure \ref{#1}}}}
\newref{sec}{refcmd={\hyperref[#1]{Section \ref{#1}}}}
\newref{as}{refcmd={\hyperref[#1]{Assumption \ref{#1}}}}
\newref{ta}{refcmd={\hyperref[#1]{Table \ref{#1}}}}
\newref{def}{refcmd={\hyperref[#1]{Definition \ref{#1}}}}
\newref{subsec}{refcmd={\hyperref[#1]{Subsection \ref{#1}}}}
\newref{ap}{refcmd={\hyperref[#1]{Appendix \ref{#1}}}}
\newref{eq}{refcmd={\hyperref[#1]{equation (\ref{#1})}}}

\newref{trm}{refcmd={\hyperref[#1]{Theorem (\ref{#1})}}}

\usepackage{footnote}
\makesavenoteenv{tabular}
\makesavenoteenv{table}

\usepackage{float}

\makeatother

\begin{document}

\title{{\Large  Stabilizing Congestion in Decentralized Record-Keepers}\footnote{We would like to thank Joseph Bonneau, Tarun Chitra, and Georgios Konstantopoulos for helpful feedback and discussions.}} 
\author{Assimakis Kattis\footnote{Courant Institute for Mathematical Sciences, New York University, kattis@cs.nyu.edu} \and Fabian Trottner\footnote{Department of Economics, Princeton University, trottner@princeton.edu}}

\maketitle

\maketitle
\begin{abstract}
  We argue that recent developments in proof-of-work consensus mechanisms can be used in accordance with advancements in formal verification techniques to build a distributed payment protocol that addresses important economic drawbacks from cost efficiency, scalability and adaptablity common to current decentralized record-keeping systems.  We enable the protocol to autonomously adjust system throughput  according to a feasibly computable statistic - system difficulty. We then provide a formal economic analysis of a decentralized market place for record-keeping that is consistent with our protocol design and show that, when block rewards are zero, the system admits stable, self-regulating levels of transaction fees and wait-times across varying levels of demand. We also provide an analysis of the various technological requirements needed to instantiate such a system in a commercially viable setting, and identify relevant research directions.
\end{abstract}
\newpage

\section{Introduction}

Blockchain technology (\cite{Nakamoto2008}) offers a decentralized alternative to traditional
record-keeping systems. Rather than relying on known central parties,
transactions are verified publicly by a decentralized network of record-keepers
(``miners'') that perform computationally expensive tasks (``proofs'').
As anyone can choose to become a record-keeper in the system, blockchain
technology offers a potential solution to well-known limitations of
competition among centralized ledgers (\cite{AbadiBrunnermeier2018}).
From an economic perspective, competition is essential to disciplining
rent-extraction by record-keepers.

However, existing blockchain technologies suffer from important economic limitations. The underlying proof-of-work (PoW) consensus mechanism that ensures correctness/trust poses inefficient levels
of cost that grow linearly with the size of the system (\cite{AbadiBrunnermeier2018,Budish2018,Thum2018}). 
In terms of \emph{scalability}, binding technological capacity constraints
limit the ability of current systems to process new record entries
at levels that would make them suitable for widespread use. Further, unlike  standard decentralized economic markets, blockchains have the property that the level of supply (i.e. number of transactions verified per block) is not determined by profit maximizing miners, but, instead, is constant and set at the beginning of the protocol. In terms of \emph{adaptability}, such systems lack the ability to adjust throughput (i.e. the supply) to the current level of demand (\cite{HubermanLeshnoMoaellemi2019,BasuEasleyOharaMaureen2019}).
This inability to adapt the supply side of the system limits the capacity of prices (transaction fees) to efficiently regulate the market for record-keeping and leads to high levels of congestion, transaction fees and system costs as demand increases.

In this paper, we study the properties of record-keepers instantiated using succinct proof systems and show positive results in terms of both scalability and adaptability. In this model, the miners are responsible for producing proofs that a set of records (or transactions) provided to them by end-users is valid, and to add these to an append-only ledger while maintaining consensus. Such an approach generalizes the Bitcoin model, in that end-users can request a record of any truth claim, which could be the correct execution of a smart contract or indeed the verification of arbitrary computation. This is a framework in which the proof of block validity and the consensus algorithm can be combined into one process (\cite{cryptoeprint:2020:190}). This connection provides a link between miners' computational work in generating a proof (and thus a block) and overall system capacity, or the number of transactions that can be verified in a given block (also referred to as \textit{predicate size}). This is done while maintaining the same model and security properties as Nakamoto (\cite{Nakamoto2008}) consensus.

This setup has the key economic implication that profit incentives to increase computational cost can be linked to overall system throughput (i.e. generating a proof of a block verifying more transactions takes substantially longer). Since the computational work miners perform will be substantial at commercial throughput levels, any practical instantiation of such a system would also require the distribution of proof computation among many machines - similar to the `mining pool' model in Bitcoin. We outline the necessary engineering steps for embedding such a system into a protocol that distributes proof computation and in which proof generation also acts as PoW. The core economic implication of a system with both features - consensus and proof distribution - is that adjustments in predicate size and thus block capacity serve as a way of incentivizing miners to incur higher computational cost in exchange for higher transaction fee revenue.

Our work provides the following contributions: 
\begin{enumerate}
    \item  An integrated, empirically relevant equilibrium framework permitting an  analysis of the determinants of transaction fees and (entry) behavior of mining pools in a decentralized market for record-keeping under dynamic adjustments of system throughput according to arbitrary update rules.  
    
    \item A formal analysis of the dynamic equilibrium behavior induced by various update rules. Specifically, we derive an update rule that induces constant equilibrium transaction fees under arbitrary levels of demand (i.e. equilibrium elasticity of fees with respect to demand is zero due to offsetting throughput increases).
    
    \item An outline of the necessary and sufficient design conditions for such systems to be practically instantiatable.
\end{enumerate}

We model the behavior of miners and users in order to reason about the change in equilibrium transaction fees subject to demand shocks. We characterize the equilibrium behavior of both miners and system users and use our model to formally derive the optimal adjustment of throughput to changes in difficulty under which changes in demand lead to efficient changes in system cost. We prove that under the derived optimal rule, and when block rewards are zero, changes in demand lead to balanced growth of throughput and mining revenues, while transaction fees and wait times remain stable.

 Our analysis yields two conditions under which stable user fees and wait times across changing levels of demand can be sustained. First, block rewards need to equal zero. Only if block rewards are equal to zero, changes in the equilibrium entry and computational behavior of miners and thus difficulty can serve as a sufficient statistic for changes in demand. Although important at the early stages of network growth, positive block rewards act as a subsidy in this context and decrease market efficiency. Second, predicate size updates need to be gradual enough so that sudden large changes in demand do not destabilize the system. Up to some maximal tolerated change in demand, an update rule which takes into account changes in difficulty to alter predicate size proportionately to the shift in demand would admit stability in its average transaction fees and wait times. Note that this admits a trade-off between the magnitude of the change in demand that can be tolerated and the time it takes for the system to converge to equilibrium.
 
 We then analyze the conceptual and quantitative implications of our theoretical analysis for the practical instantiation of a scalable, self-regulating record-keeping system. We first look at related work on distributed proof generation (\cite{wu2018dizk}) and compare the functional equivalence of our model's assumptions to the observed data. This provides design specifications and proof system properties that can be used to construct payment systems with stability guarantees. We then detail the various areas of engineering relevance required in explicitly instantiating such a protocol. We outline the types of properties that the underlying system should possess, look at the relevant technical literature, and provide directions for future work that will need to be done along with sufficient conditions for its implementation.
 
 Jointly, our results provide theoretically grounded guidance toward the practical construction of a self-regulating, cost-efficient and decentralized payment system with permissive entry structure and high security guarantees. Our design outlines how current technical advances in the computer science literature can address significant economic limitations of current implementations.  We show how economic analysis can help understand how technical features of the design can be leveraged so as to align the incentives of active record-keepers with the requirements for cost-efficiency.  Further, our economic analysis provides concrete, measurable performance targets, which we believe will be useful in guiding future efforts to build distributed payment systems that scale in an efficient, decentralized manner to the demands of end-users.
 
  \section{Preliminaries}\label{sec:design}

 \subsection{Proof of Work}
Many blockchain protocols ensure security by requiring miners to publicly verify energy consumption (known as proof-of-work or PoW) in order to gain the right to write transactions on the blockchain. PoW is required in order to prevent denial-of-service attacks on the system: the potential for malicious entities
to flood the system with cheaply produced `valid' blocks can rapidly overwhelm honest
participants and hamper the network's ability to process valid transactions. PoW was first
introduced as a solution to the related problem of e-mail spam filtering (\cite{dwork1992pricing}), and provides the security guarantees against such attacks in Bitcoin (\cite{back2002hashcash}). By imposing a high energy cost for block creation, PoW makes it prohibitively expensive to produce many valid blocks and overwhelm the system, as doing so carries a non-trivial economic burden.

In the subsequent sections, we model this security constraint as the requirement that as demand $\lambda$ increases, some sufficient statistic (which in the case of Bitcoin is difficulty $d$) also causes a proportional increase in miner costs.

\subsection{Difficulty}

Since the network is
decentralized, the addition of a new block to the blockchain takes some time to propagate
among all parties. This is due to the latency inherent in distributed networks, and imposes
a trade-off between the size of the block propagated and the time taken for the rest of the
network to receive the updated block. This limitation implies that the frequency with which
the system accepts new blocks (`block frequency') needs to be sufficiently low so that each
new block has enough time to reach the whole system. This is necessary to prevent forks (a `split' in the chain), and in practice results in a trade-off between block size and frequency.
Note that this also implies an upper bound on improving throughput through modifications to block
size and/or frequency.
  
Since the total number of miners is an endogenous variable that varies over time, the protocol
needs a way to ensure that a PoW solution is provided at a fixed frequency (on average)
over \textit{all} participants in the network. This is achieved by requiring the PoW puzzle to satisfy
a `difficulty' requirement, which would mean that the expected time taken to generate a
solution can be calibrated according to a difficulty parameter that can be changed over time.
Through the use of time-stamps, the Bitcoin PoW algorithm keeps track of the average block
frequency in the last $2016$ blocks and alters difficulty up or down accordingly. By performing
this adjustment proportionately to the deviation of the block frequency from a hard-coded
value (10 minutes/block in Bitcoin), the system `self-corrects' the difficulty of the PoW
puzzle to ensure constant block frequency. This mechanism is crucial to ensuring that
blocks in Bitcoin are found every 10 minutes, so that the previous blocks had enough time
to propagate in the network and the threat of forks is minimized. In the subsequent analysis, we model the block frequency as a random variable $\mu$ (the block arrival rate) with a constant expectation (the `block time') mirroring the above technological constraints.

\subsection{Difficulty Adjustment}

The underlying system consists of a set of users and miners, who each are responsible for requesting and providing proofs from and to the system respectively. Everyone has access to the underlying state of the system, which consists of a sequence of valid transactions that have been processed into the system (the blockchain) along with their respective proofs of correctness: proofs that each block added to the chain is valid.

Users submit transactions (or `records') to miners along with an associated transaction fee. The miners in turn are responsible for adding these transactions to some block which is then appended to the blockchain. This happens if the miner succeeds in generating a PoW solution that also acts as the proof of validity of all transactions in the block. We denote by $S$ the total number of transactions verified in the given block. In order for a miner to win the right to add the next block (and hence collect the corresponding transaction fees and block reward), they need to win the PoW game, which is a puzzle yielding a valid block in time proportional to the underlying difficulty parameter. This difficulty parameter functions equivalently to that in Bitcoin: through periodic updates, the difficulty parameter is adaptively changed to ensure constant block frequency. In this model, only miners are required to keep the whole blockchain (or some kind of `state' data structure) in order to efficiently generate new blocks and proofs. End-users should be able to efficiently check provided proofs in order to find the correct chain with minimal bandwidth requirements.

\section{Economic Model}

 \subsection{Overview}
   We model a decentralized market for throughput where transactions are verified in batches/blocks that arrive at a random, exogenous  rate $\mu$ in each unit of time. The chart displayed in \figref{flow_chart} represents the key forces in our model. 
   
   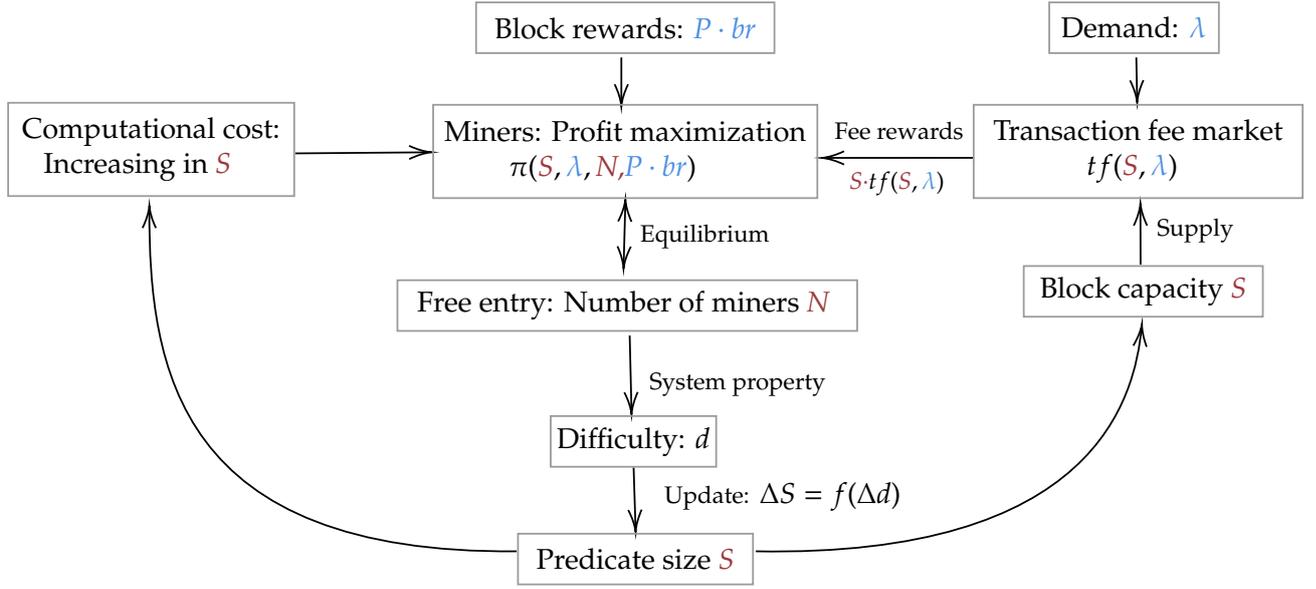
\begin{figure}[h]   
     \caption{A Model  of the Decentralized Market for Throughput}\label{fig:flow_chart}
     \scalebox{0.9}{\tikzset{every picture/.style={line width=0.75pt}} 

\begin{tikzpicture}[x=0.85pt,y=0.85pt,yscale=-1,xscale=1]

\draw    (380.85,283.18) .. controls (547.31,286.3) and (566.3,207.59) .. (569.7,172.74) ;
\draw [shift={(569.85,171.18)}, rotate = 455.04] [color={rgb, 255:red, 0; green, 0; blue, 0 }  ][line width=0.75]    (10.93,-3.29) .. controls (6.95,-1.4) and (3.31,-0.3) .. (0,0) .. controls (3.31,0.3) and (6.95,1.4) .. (10.93,3.29)   ;
\draw    (415.18,89.18) -- (487.09,89.15) ;
\draw [shift={(413.18,89.18)}, rotate = 359.98] [color={rgb, 255:red, 0; green, 0; blue, 0 }  ][line width=0.75]    (10.93,-3.29) .. controls (6.95,-1.4) and (3.31,-0.3) .. (0,0) .. controls (3.31,0.3) and (6.95,1.4) .. (10.93,3.29)   ;
\draw    (84.09,115.17) .. controls (84.04,192.84) and (97.53,284.72) .. (263.85,283.18) ;
\draw [shift={(84.09,112.82)}, rotate = 90.23] [color={rgb, 255:red, 0; green, 0; blue, 0 }  ][line width=0.75]    (10.93,-3.29) .. controls (6.95,-1.4) and (3.31,-0.3) .. (0,0) .. controls (3.31,0.3) and (6.95,1.4) .. (10.93,3.29)   ;
\draw    (155.51,86.95) -- (220.18,86.42) ;
\draw [shift={(222.18,86.4)}, rotate = 539.52] [color={rgb, 255:red, 0; green, 0; blue, 0 }  ][line width=0.75]    (10.93,-3.29) .. controls (6.95,-1.4) and (3.31,-0.3) .. (0,0) .. controls (3.31,0.3) and (6.95,1.4) .. (10.93,3.29)   ;
\draw    (315.09,39.65) -- (315.09,61.82) ;
\draw [shift={(315.09,63.82)}, rotate = 270] [color={rgb, 255:red, 0; green, 0; blue, 0 }  ][line width=0.75]    (10.93,-3.29) .. controls (6.95,-1.4) and (3.31,-0.3) .. (0,0) .. controls (3.31,0.3) and (6.95,1.4) .. (10.93,3.29)   ;
\draw    (317.04,111.65) -- (316.22,141.76) ;
\draw [shift={(316.17,143.76)}, rotate = 271.56] [color={rgb, 255:red, 0; green, 0; blue, 0 }  ][line width=0.75]    (10.93,-3.29) .. controls (6.95,-1.4) and (3.31,-0.3) .. (0,0) .. controls (3.31,0.3) and (6.95,1.4) .. (10.93,3.29)   ;
\draw [shift={(317.09,109.65)}, rotate = 91.56] [color={rgb, 255:red, 0; green, 0; blue, 0 }  ][line width=0.75]    (10.93,-3.29) .. controls (6.95,-1.4) and (3.31,-0.3) .. (0,0) .. controls (3.31,0.3) and (6.95,1.4) .. (10.93,3.29)   ;
\draw    (567.09,39.65) -- (567.48,59.84) ;
\draw [shift={(567.51,61.84)}, rotate = 268.92] [color={rgb, 255:red, 0; green, 0; blue, 0 }  ][line width=0.75]    (10.93,-3.29) .. controls (6.95,-1.4) and (3.31,-0.3) .. (0,0) .. controls (3.31,0.3) and (6.95,1.4) .. (10.93,3.29)   ;
\draw    (569.1,113.65) -- (569.18,141.84) ;
\draw [shift={(569.09,111.65)}, rotate = 89.84] [color={rgb, 255:red, 0; green, 0; blue, 0 }  ][line width=0.75]    (10.93,-3.29) .. controls (6.95,-1.4) and (3.31,-0.3) .. (0,0) .. controls (3.31,0.3) and (6.95,1.4) .. (10.93,3.29)   ;
\draw    (319.17,176.76) -- (320.05,213.65) ;
\draw [shift={(320.09,215.65)}, rotate = 268.63] [color={rgb, 255:red, 0; green, 0; blue, 0 }  ][line width=0.75]    (10.93,-3.29) .. controls (6.95,-1.4) and (3.31,-0.3) .. (0,0) .. controls (3.31,0.3) and (6.95,1.4) .. (10.93,3.29)   ;
\draw    (321.09,241.82) -- (322.03,271.65) ;
\draw [shift={(322.09,273.65)}, rotate = 268.2] [color={rgb, 255:red, 0; green, 0; blue, 0 }  ][line width=0.75]    (10.93,-3.29) .. controls (6.95,-1.4) and (3.31,-0.3) .. (0,0) .. controls (3.31,0.3) and (6.95,1.4) .. (10.93,3.29)   ;

\draw  [color={rgb, 255:red, 155; green, 155; blue, 155 }  ,draw opacity=1 ]  (244,12.5) -- (390,12.5) -- (390,37.5) -- (244,37.5) -- cycle  ;
\draw (317,25) node   [align=left] {Block rewards: \textcolor[rgb]{0.74,0.06,0.88}{$\displaystyle \textcolor[rgb]{0.29,0.56,0.89}{P\cdot br}$}};
\draw  [color={rgb, 255:red, 155; green, 155; blue, 155 }  ,draw opacity=1 ][line width=0.75]   (524.5,12.5) -- (607.5,12.5) -- (607.5,37.5) -- (524.5,37.5) -- cycle  ;
\draw (566,25) node   [align=left] {Demand: \textcolor[rgb]{0.74,0.06,0.88}{$\displaystyle \textcolor[rgb]{0.29,0.56,0.89}{\lambda }$}};
\draw  [color={rgb, 255:red, 155; green, 155; blue, 155 }  ,draw opacity=1 ]  (487.5,63) -- (648.5,63) -- (648.5,109) -- (487.5,109) -- cycle  ;
\draw (568,86) node   [align=left] {Transaction fee market\\ \ \ \ \ \ \ \ \ \ \ \ \ \ $\displaystyle \textcolor[rgb]{0,0,0}{tf}\textcolor[rgb]{0,0,0}{(}\textcolor[rgb]{0.62,0.21,0.23}{S} ,\textcolor[rgb]{0.29,0.56,0.89}{\lambda })$};
\draw  [color={rgb, 255:red, 155; green, 155; blue, 155 }  ,draw opacity=1 ]  (223,63) -- (411,63) -- (411,109) -- (223,109) -- cycle  ;
\draw (317,86) node   [align=left] {Miners: Profit maximization\\ \ \ \ \ \ \ \ \ \ $\displaystyle \pi (\textcolor[rgb]{0.62,0.21,0.23}{S} ,\textcolor[rgb]{0.29,0.56,0.89}{\lambda } ,\textcolor[rgb]{0.62,0.21,0.23}{N,}\textcolor[rgb]{0.29,0.56,0.89}{P\cdot br})$};
\draw  [color={rgb, 255:red, 155; green, 155; blue, 155 }  ,draw opacity=1 ]  (205.5,149.5) -- (430.5,149.5) -- (430.5,174.5) -- (205.5,174.5) -- cycle  ;
\draw (318,162) node   [align=left] {Free entry: Number of miners $\displaystyle \textcolor[rgb]{0.6,0.23,0.23}{N}$ };
\draw  [color={rgb, 255:red, 155; green, 155; blue, 155 }  ,draw opacity=1 ]  (280.5,216.5) -- (361.5,216.5) -- (361.5,241.5) -- (280.5,241.5) -- cycle  ;
\draw (321,229) node   [align=left] {Difficulty: $\displaystyle d$};
\draw  [color={rgb, 255:red, 155; green, 155; blue, 155 }  ,draw opacity=1 ]  (264.5,274.5) -- (379.5,274.5) -- (379.5,299.5) -- (264.5,299.5) -- cycle  ;
\draw (322,287) node   [align=left] {Predicate size $\displaystyle \textcolor[rgb]{0.62,0.21,0.23}{S}$};
\draw  [color={rgb, 255:red, 155; green, 155; blue, 155 }  ,draw opacity=1 ]  (15,62) -- (155,62) -- (155,108) -- (15,108) -- cycle  ;
\draw (85,85) node   [align=left] {Computational cost:\\ \ \ \ Increasing in $\displaystyle \textcolor[rgb]{0.62,0.21,0.23}{S}$};
\draw  [color={rgb, 255:red, 155; green, 155; blue, 155 }  ,draw opacity=1 ]  (512.01,142.5) -- (629.01,142.5) -- (629.01,167.5) -- (512.01,167.5) -- cycle  ;
\draw (570.51,155) node   [align=left] {Block capacity $\displaystyle \textcolor[rgb]{0.62,0.21,0.23}{S}$};
\draw (451,75.64) node   [align=left] {{\footnotesize Fee rewards}};
\draw (447,101.64) node   [align=left] {{\footnotesize $\displaystyle \textcolor[rgb]{0.62,0.21,0.23}{\ \ S\cdot } tf(\textcolor[rgb]{0.62,0.21,0.23}{S} ,\textcolor[rgb]{0.29,0.56,0.89}{\lambda })$}};
\draw (356,128) node   [align=left] {{\footnotesize Equilibrium}};
\draw (596,125) node   [align=left] {{\footnotesize Supply}};
\draw (372,201) node   [align=left] {{\footnotesize System property}};
\draw (394,256) node   [align=left] {{\footnotesize Update:} $\displaystyle \Delta S=f( \Delta d)$};

\end{tikzpicture}}
     
   \end{figure}
   Miners compete for the right to verify an incoming set of transactions in exchange for transaction fees and block rewards. They  choose the profit-maximizing number of workers among which to distribute proofs to simultaneously verify block correctness and  generate a PoW solution. In equilibrium all miners include the transactions offering the highest fees, up to block capacity $S$.\footnote{In practice, miners may choose to include different transactions in each block. As we assume that miners observe the same pending transactions at any given point in time and miners have incentives to include transactions with higher fees, miners all aim to verify the same block of transactions.} The probability that  a miner gains the right to write the next block is proportional to the number of computations that she performs. The profitability of miners thus depends on the total number $N$ of active miners. Free entry requires that all miners make zero profits in expectation, which pins down the equilibrium number of  miners $N$.
   
    Predicate size $S$ corresponds to the number of transactions that can be verified in each new block. Users arrive at the pool of pending transactions at rate $\lambda$ and offer fees for transaction verification. The determination of transaction fees results from a Vickrey-Clarke-Groves auction, where users optimally weigh the cost of higher fees against a reduction in their expected wait time. In equilibrium, average transaction fees $tf$  depend on both block capacity (and thus predicate size $S$), and on the level of demand $\lambda$.  
   
   Unlike  standard economic markets, the level of supply - system throughput/block capacity $S$ - cannot be directly chosen by  miners, but is instead set by the  protocol. The inability of miners to choose both their computational cost \textit{and} throughput poses an economic market failure that leads to inefficient adjustments of system cost to changes in demand. To address this limitation, we model a system that implements periodic updates to predicate size in accordance with an ex-ante specified update rule. The purpose of this updating rule is to grow throughput capacity in response to changes in demand.  Our framework highlights that equilibrium system difficulty -- the total proof computations performed by all miners -- responds to changes in demand, and, therefore, serves as a natural candidate statistic to design such an updating rule.
   
   Our framework allows us to analyze how key outcomes -- transaction fees, computational cost of the system, entry of miners -- respond dynamically to changes in demand.  A change in demand $\lambda$ induces changes in transaction fees offered by users, which changes difficulty by affecting the equilibrium number and size of active mining pools. The subsequent update in predicate size  changes block capacity and therefore transaction fees, inducing in turn a change in mining profitability,  the equilibrium number of miners, and, thus, difficulty.  A steady state is reached when difficulty and thus predicate size stabilize after the initial change in demand.
   
   In this section, we focus on characterizing the problem of economic agents, and on defining a decentralized equilibrium. In the next section, we apply the framework to investigate the dynamic behavior of the system under various updating rules for predicate size.  

  \subsection{Miners} 
  
  \subsubsection{The relation between hash rate, predicate size, and pool size}
  Miners compete for the chance to mine the next arriving block and collect the associated mining rewards. To do so, miners distribute proof computation corresponding to predicate size $S$ to pools and provide PoW in terms of proofs per second $H$ to the system. The relationship between  proofs per second $H$ for a miner employing $n$ workers to solve proofs corresponding to a predicate size $S$ is given by:
  \begin{equation}\label{eq:production}
  \log H(n;S) = \log U + \sigma \log n - \log S
  \end{equation}
  where $U$ is a technological constant.  The parameter $\sigma$ governs  the returns to pool size $n$  - the rate at which the  hash rate increases as miners increase the size of their pools.  This can be thought of as a parallelization coefficient, for which throughout the analysis satisfies $\sigma <1$.\footnote{State-of-the-art distributive systems of SNARK proofs achieve $\sigma\approx 0.9$ (\cite{wu2018dizk}).} Note that a value of $\sigma = 1$ would imply perfect conservation of work in a parallel system of $n$ workers, and as such no higher value of $\sigma$ is empirically achievable.
  
  As evident in \eqref{production}, an increase in predicate size $S$ reduces a pool's hash rate  for any given pool size.  An increase in $S$ thus poses a cost to the system as pools have to distribute proofs across more workers to maintain the same hash rate.

  \subsubsection{The Problem of Miners}
  Miners incur a cost $c_m\cdot n$ per time unit to distribute proofs across $n$ workers.\footnote{In practice, miners distribute and assemble proofs while workers solve small computational problems. The cost $c_m$ summarizes both margins.} Further, (in each period) miners pay a fixed cost $f_m$ to open and maintain a mining pool. Miners  compete for the chance to mine the next arriving block. In equilibrium, all miners include  transactions offering the highest rewards up to block capacity $S$. Benefits of mining a block arise from total transaction fee revenue $Rev$ and block rewards $br$. Transaction fees are denominated in USD, as users propose the transaction fees. We specify $Rev$ further below after having solved the problem of users. As  block rewards are exogenously set by the system, their value to miners is impacted by the exchange rate of the currency to USD, which we denote by $P$. Total rewards from mining are given by $Rev+P\cdot br$. 
  
  As in PoW, the probability that a given pool is selected to mine the next block depends on its hashrate relative to the aggregate hashrate of the system.\footnote{\cite{axiomaticPOW2019} and \cite{LeshnoStrack2019} show that  proportional selection is the only allocation rule that satisfies a set of desirable properties, e.g. anonymity, collusion-resistance, and sybil-resistance.} We denote  individual pools by $i$ and the aggregate hashrate of the system by $\mathcal{H}=\sum_{i\in\{\text{Active Pools}\}} H_i$. 
  
  Taking predicate size $S$ and the aggregate hashrate $\mathcal{H}$ as given, a miner $i$ chooses pool size $n_i$ so as to  maximize expected profits $\pi$ given by:
  \begin{equation}\label{eq:profits}
  \pi\left(n_i;S,\mathcal{H}\right)= \frac{H\left(n_i,S\right)}{\mathcal{H}}\left(Rev+P\cdot br\right)-c_m\cdot n_i-f_m,
  \end{equation}
  where $H\left(n,S\right)$ is defined in \eqref{production}. 
  
 A Nash equilibrium requires that equilibrium pool sizes are pinned down by miners' best-responses:
 \begin{equation}
 n_i^*=\arg\max_n \pi\left(n;S,\mathcal{H}\right)
 \end{equation}
 The solution to this problem is given by:\footnote{When setting $\frac{\partial\pi}{\partial n_i}=0$, we assume that miners are small in the sense that they don't internalize the effect their computations have on aggregate hash-rate $\mathcal{H}$. Our qualitative results do not crucially depend on this assumption.}
 
 \begin{equation}\label{eq:optimal_pool_size}
 n^*_i = \left(\frac{\sigma\left(Rev+P\cdot br\right)}{c_m(S/U)\mathcal{H}}\right)^{1/(1-\sigma)}.
 \end{equation}

  \subsubsection{Entry and the Equilibrium Number of Miners}
  Everybody is free to enter the system as a miner. In equilibrium,  the number of miners $N$ adjusts to ensure that upon entry, mining pools make zero profits in expectation. 
  
 Since mining pools are symmetric,  in equilibrium $n_i^*\equiv n^*$ and $\mathcal{H}=NH\left(n^*\right)$. Using \eqref{optimal_pool_size}, the optimal pool size can be shown to equal: 
  \begin{equation}\label{eq:optimal_n}
  n^*=\frac{\sigma}{c_m}\cdot\frac{Rev+P\cdot br}{N}
  \end{equation}
  The equilibrium probability that any miner is chosen to mine the next block equals $\frac{H_i}{\mathcal{H}}=\frac{H}{NH}=1/N$.  Substituting \eqref{optimal_n} into profits given by \eqref{profits} and imposing zero expected profits upon entry, the equilibrium number of miners is given by: 
  \begin{equation}\label{eq:free_entry}
  N = \frac{(1-\sigma)}{f_e}\cdot\left(Rev+P\cdot br\right)
  \end{equation}
  The equilibrium number of miners is increasing in mining revenues and decreasing in the entry cost.  By virtue of our implementation, miners that solve predicates of size $S$ and get chosen to write the next block can verify up to $S$ transactions.  Since there is entry cost and miners cannot force other miners into a coalition, miners have no incentive to form coalitions with the  intent to temper with transaction fees in order to increase profits.\footnote{For example, a large coalition may choose to not include transaction fees below a certain threshold.} Therefore, miners include transactions that offer the highest transaction fees up to the capacity limit.
  
 \subsection{Demand for Transactions and Determination of Fees}
 
 We can now describe the problem of users, which yields the equilibrium level of transaction fees. The  demand side of the model closely follows \cite{HubermanLeshnoMoaellemi2019}. Our model is distinct in one key dimension. The adaptive nature of our protocol implementation implies that block capacity $S$ and thus equilibrium system throughput is not a primitive, but rather an equilibrium outcome.
 
 Users have single transactions and are heterogeneous in the  cost that verification delay poses for them. New users enter the pool of pending transactions at a constant random rate per time unit and offer transaction fees for service. Knowing that a higher posted fee improves the chances of being included in a block sooner, users willingness to pay arises out of the opportunity cost of delay.

 \subsubsection{Users}
 Users arrive at Poisson rate $\lambda$ at each point in time. $\lambda$ parametrizes demand intensity or the rate at which the pool of pending transactions grows over time.  An arriving user is identified by a wait cost $c$, which is drawn from a continuous cumulative density function  $F(c)$ with domain $C\equiv\left[0,\bar{c}\right]\subseteq\mathbb{R}_0^+$ and probability density function $f(c)$. The expected benefit of using the system to a user endowed with wait cost $c$ and  offering  transaction fee $tf$ is given by:
 
 $$U\left(tf;W,c\right)=\nu - tf - c W\left(tf;G\right),$$ 
 
 where $\nu>0$ governs the value of a verified transaction to a user\footnote{In practice, users have an outside option and only participate if their utility from doing so exceeds this outside option. Throughout the analysis, we  assume that $\nu$ is sufficiently high so that users are willing to participate.} and $W\left(tf,G\right)$ is the expected wait time for a user that posts  fee $tf$. The Wait-time $W$ depends on the (endogenous) distribution of  transaction fees posted by all users $G(tf)$  in the system and is thus itself an equilibrium object.\footnote{The assumption is that users know the distribution of wait times $F$ and are assumed to correctly anticipate the optimal behavior of others. However, users do not observe the current pool of pending transactions when submitting fee offers.}

 \subsubsection{The Dependence of Fees on Wait Time}
 Each active user posts a fee $tf$ so as to maximize utility, weighing the cost of posting higher fees against the benefit of lowering the expected wait time. Standard arguments imply that equilibrium transaction fees are a continuous, monotonous function of user wait-cost $tf\left(c\right)$ satisfying  $tf(0)=0$ and  $tf'(c)>0$. Monotonicity implies that $G\left(tf(c)\right)=F\left(c\right).$  Equilibrium transaction fees $tf(c)$ solve the first order  condition $W'\left(tf|G\right)=-\frac{1}{c}$ which is equivalent to the following differential equation:
 $$\tilde{W}'\left(c|F\right)cf(c)=tf'(c),$$ 
 where $\tilde{W}$ directly maps wait cost into wait times and satisfies $-W'(tf|G)=W'(c|F)$.\footnote{Note that $G\left(tf(c)\right)=F(c)$ implies $G'(b(c))b'(c)=f(c)$. Thus $W'(tf|G)=W'(tf|G)f(c)/b'(c).$} Integrating and using the property that users with no delay cost post zero fees, transaction fees can be fully expressed in terms of equilibrium wait times:
 \begin{equation}\label{eq:fees}
 tf(c)=\int_0^{c}\tilde{c}f(\tilde{c})\tilde{W}'\left(\tilde{c}|F\right)d\tilde{c}.
 \end{equation}
 
Inspecting \eqref{fees}  reveals that equilibrium fees have the externality correcting property  inherent to the Vickrey-Groves-Clark (VCG) auction mechanism.\footnote{An externality arises as each user delays the  verification of the transactions of other users. Despite the fact that users do not internalize the cost that they pose on others, the VCG auction mechanism incentivices "truthful bidding", implying that the users with higher valuations submit higher bids.} For our purposes, the fact that offered transaction fees are socially optimal has the implication that none of the potential cost-efficiencies exhibited by the decentralized equilibrium arise from suboptimal user behavior.

 \subsubsection{Equilibrium Wait Time}
 Blocks  can process up to $S$ transactions per time period  and arrive randomly at Poisson rate $\mu$.  The equilibrium wait time naturally depends on system congestion, which in our context is defined as $\rho\equiv \lambda/S\mu$ - the average ratio of demand  $\lambda$ to the average number of transactions that can be processed per time unit. $\rho<1$ implies that while users may experience delays, all transactions will eventually be processed by the system.
 
 To characterize the equilibrium wait time, we build on results derived in \cite{HubermanLeshnoMoaellemi2019}.   The authors show that for sufficiently high levels of block capacity, wait times are fully characterized by system congestion $\rho$.\footnote{\cite{HubermanLeshnoMoaellemi2019} show wait times  are appropriately approximated for $S=20$. Block capacity in Bitcoin is equal to 2000.  As outlined earlier, we expect our proposed protocol implementation  to be capable of handling a volume of transactions per block that outperforms Bitcoin by magnitudes. We therefore feel comfortable to use  the characterization of wait times in \eqref{wait_time} for the remainder of the analysis.} We restate this result in the following theorem. 
 
 \begin{theorem}[\cite{HubermanLeshnoMoaellemi2019}]

 \begin{enumerate}
 \item For any $\rho\equiv\lambda/(\mu S)\in\left(0,1\right)$ the equilibrium wait time for a user  with delay cost $c$, $W(c;S,\mu,\lambda)$ is given by: 
 
 \begin{equation}\label{eq:wait_time}
 W(c;\rho)=\frac{\mu^{-1}}{1-\left(1+\alpha\left(\rho\bar{F}(c)\right)\right)e^{-\alpha(\rho\overline{F}(c))}},
 \end{equation}
  where $\overline{F}(c)\equiv 1-F(c)$ and $\alpha(x)$ is the real root of the algebraic equation $e^{-\alpha}+x\alpha-1=0$.
  \item  Equilibrium wait time is increasing and convex in congestion $\rho$ and decreasing in user wait cost $c$.
  \end{enumerate}
 \end{theorem}
 \begin{myproof1}
 See Lemma 14 in \cite{HubermanLeshnoMoaellemi2019}.
 \end{myproof1}
 
  The  relationship between equilibrium  wait times and optimal bidding behavior implies that  users with higher wait cost offer higher fees, and wait less for transaction verification. Higher levels of system congestion increase wait times - and therefore transaction fees - for all users.  Theorem 1 implies that wait times - and, therefore, transaction fees - are fully characterized by system congestion $\rho$,\footnote{In principle, wait-times depend on block capacity $S$ in addition to congestion. On the one hand for a given transaction, the inclusion in  a block depends on how many pending transactions have accumulated at the time  block is generated. This depends on system congestion $\rho$.  On the other hand, the inclusion of the transaction might also depend on new transactions with higher priority that might arrive before the next block. The  arrival of higher priority users creates significant randomness and thus affects  wait times only if blocks have sufficiently small capacities such that new arriving  transactions can significantly change the composition of included transactions. For sufficiently high $S$, however, wait times are nearly independent from this source of variation.}  implying that equilibrium    stability in fees  is equivalent to stability in system congestion $\rho$.  
  
 \subsubsection{Equilibrium Transaction Fees}
 To summarize the optimal bidding behavior across all users, Theorem 1 and \eqref{fees} allow to derive  the average equilibrium fee revenue  $Rev$ per unit of time, given block capacity $S$ and system congestion $\rho$: 
 
 \begin{equation}\label{eq:fee_revenue}
 Rev\left(S,\rho\right) = S\cdot\int_0^{\bar{c}} tf(c)dF(c)=S\cdot\rho\int_0^{\bar{c}} \left(\bar{F}(c)-cf(c)\right)W\left(c;\rho\right)dc\equiv S\cdot \psi\left(\rho\right).
 \end{equation}
 
 $\psi(\rho)$ in \eqref{fee_revenue} is  the average level of transaction fees per unit of time that users are willing to pay for service.\footnote{As we characterize steady states in a fee market where there is  a constant  flow of transactions in and out of a system that eventually verifies all transactions and where all users participate, the integral is  taken over the entire domain of $c$.} As $Rev$ summarizes the optimal bidding behavior of users,  demand side effects stemming from changes in either block capacity $S$ or demand $\lambda$ are fully captured by its properties.
 
 The key focus of our analysis is to characterize the  system's ability to induce incentive compatible increases in throughput so as to absorb an increase in demand $\lambda$ at constant average prices. The responsiveness of average fees $\psi\left(\rho\right)$ with respect to changes in demand $\lambda$ - or generally changes in congestion $\rho$ - is key for this purpose.  To this end, we define the elasticity of  $\psi(\rho)$ in \eqref{fee_revenue} with respect to $\rho$ as $\varepsilon\left(\rho\right)\equiv \frac{\partial \log \psi\left(\rho\right)}{\partial \log \rho}\equiv \frac{\rho}{\psi(\rho)}\frac{\partial \psi(\rho)}{\partial \rho}$.  $\varepsilon\left(\rho\right)$ captures the percentage increase in equilibrium fees in response to a percentage increase in congestion $\rho$. The following theorem characterizes key properties of  $\varepsilon\left(\rho\right)$ and in particular shows that it admits a uniform upper bound, independently of the underlying distribution of wait cost $F(c)$. The proof is delegated to \apref{proofs}.

 \begin{theorem}
 The elasticity of equilibrium average transaction fees $\psi\left(\rho\right)$ with respect to system congestion $\rho$, $\varepsilon\left(\rho\right)\equiv \frac{\partial \log \psi\left(\rho\right)}{\partial \log \rho}$,  is bounded below by 1,  increasing and convex in $\rho$. Further, there exists  $\overline{\varepsilon}\left(\rho\right)$ such that $\varepsilon\left(\rho\right)\leq \overline{\varepsilon}\left(\rho\right)$ for all $\rho\in(0,1)$ and any distribution of user wait cost $F(c)$. 
 \end{theorem}
 
 \figref{elasticity_fees} plots the uniform upper bound on the elasticity of average fees with respect to changes in congestion derived in Theorem 2. 
  \begin{figure}
\caption{Upper Bound on the  Congestion Elasticity of Equilibrium Average Fees}
\label{fig:elasticity_fees}
    \centering
    \includegraphics[scale=0.45]{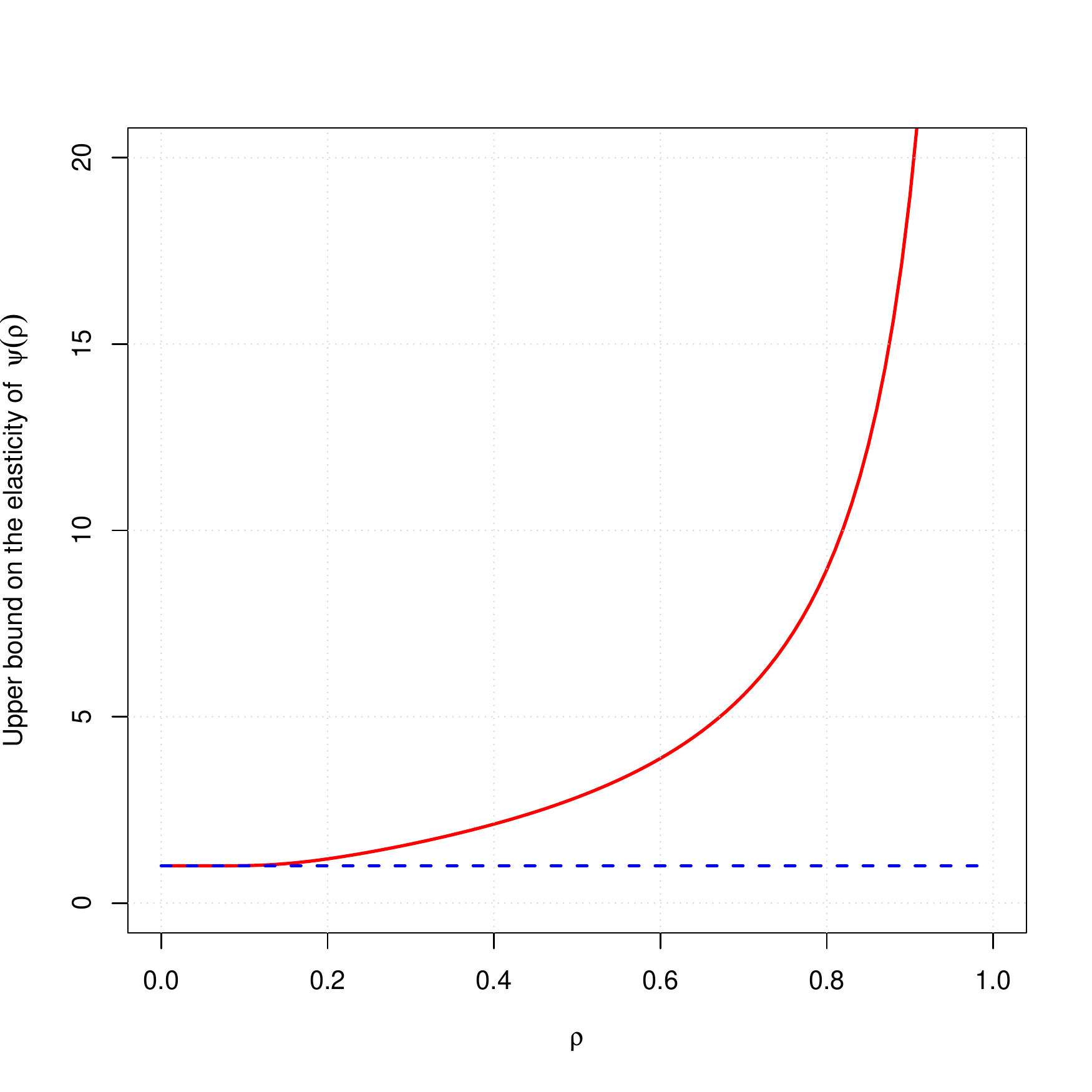}
    
    \footnotesize{Notes: This figure plots the uniform bound on the elasticity of equilibrium average fess given in  \eqref{fee_revenue} with respect  to system congestion $\rho=\lambda/\mu S$ as a function of $\rho$ on the x-axis.}
\end{figure}

 Total mining fees $Rev$ summarize the demand-side in the market. Jointly with the free entry condition for miners, $Rev$ allows us to derive  the equilibrium number of miners $N$ for a given predicate size $S$. 
 
 \subsection{Equilibrium for a Fixed Predicate Size $S$}
 
 \subsubsection{Equilibrium Number of Miners}
 
 A partial equilibrium for a given predicate size $S$ and demand level $\lambda$ is defined by the condition that miners make zero profits in expectation upon entry.\footnote{While this equilibrium is technically  a steady state - owing to the continuous entrance and exit of users - we effectively treat it as a static equilibrium condition.}   Using \eqref{free_entry} and the characterization of total mining fee revenues, the equilibrium number of miners $N$ at predicate size $S$, demand $\lambda$ and nominal block rewards $PBr$ is given by: 
 \begin{equation}\label{eq:partial_eq}
 N\left(S,\lambda,Pbr\right)=\frac{1-\sigma}{f_e}\cdot\left(S\psi\left(\rho\right)+Pbr\right).
 \end{equation}
 
   Predicate size $S$, demand $\lambda$ and nominal block rewards $PBr$ summarize the profit incentives of miners. Adjustments in predicate size $S$, therefore, change the profit incentives and entry behavior of miners. Cost-efficiency in a decentralized equilibrium crucially depends on whether adjustments of predicate size $S$ can be optimally calibrated so as to induce incentive-compatible equilibrium behavior of miners that is \textit{as if}  miners internalized the externality that they pose on others through Nakamoto consensus. To build intuition for our later results, it is useful to briefly pause to discuss how adjustments in predicate size $S$ shape profit incentives of miners through the condition in \eqref{partial_eq}. 
 
 An increase in demand $\lambda$ initially increases mining fee rewards through an increase in equilibrium user fees $\psi(\rho)$. This increases entry, the equilibrium number of miners, and the computational cost of the system without adding benefits for system capacity or users. To restore cost-efficiency, a subsequent increase in predicate size $S$ has to  direct additional cost incurred by the system toward increases in throughput. Our results imply that competition by miners ensures that this is achieved by additional entry and lower hash-rates for each active pool. 
 
 \subsubsection{Difficulty}
 
 As opposed to for example Bitcoin, difficulty in our model depends only on the total number of proofs  computed per time unit.\footnote{This is because the solution of a single proof commits to the nonce before distributing computation to the pool. This means that, unlike the distribution of double-SHA computations in Bitcoin (that effectively break up nonce generation among pool participants), here the nonce is set by the operator, and all computation forwarded to workers is nonce-specific.} As pools are symmetric, solely distribute one proof computation among pool participants and operate at hashrate $H$, system difficulty is given by:
 \begin{equation} \label{eq:difficulty}
  d=\log\left(\mu\times H\times N\right).
 \end{equation}
 
Using the equilibrium expression for the number of miners and the solution to the miner's problem, difficulty can be expressed as:\footnote{Note that the solution to the miners' problem implies that equilibrium hashrates are given by $U\left(\frac{\sigma f_m}{(1-\sigma)c_m}\right)^\sigma/S$.}
 
 \begin{equation} \label{eq:difficulty_1}
  d=\log\left(\alpha\times \frac{S\psi(\rho)+P\cdot br}{S}\right),
 \end{equation}
 where $\alpha\equiv \mu U\left(\frac{\sigma}{c_m}\right)^{\sigma}\left(\frac{1-\sigma}{f_e}\right)^{1-\sigma}$ is a constant. 
 
 The inability of miners to autonomously adjust cost \textit{and} throughput upwards as demand rises  causes  rising levels of system congestion, resulting in high wait times and fees for users, and  costly increases in  mining activity at no increase in throughput as demand rises. In light of this significant economic limitation, a core innovation of our protocol design is to provide an "updating rule" for predicate size $S$ and therefore block capacity that correctly infers changes in demand.  While demand $\lambda$ is not directly observable,   \eqref{difficulty_1} highlights that difficulty is a sufficient statistic for changes in demand when block rewards are zero. When block rewards are zero, changes in equilibrium difficulty are proportional to changes in average transaction fees, which, in turn, depend solely on congestion $\rho$.

 \subsection{Dynamic Throughput Adjustment}
 We endow the system with the ability to periodically adapt throughput $S$ according to a pre-specified updating rule.
 
 \begin{definition}[Law of Motion for Throughput] Let $\mathbf{X}$ be a vector of  statistics. A law of motion for throughput $S$ is a function $f$ such that:  \begin{equation}\label{eq:lom_gen}
 S_{t+1} = f\left(S_t,S_{t-1},...;\mathbf{X}_t,\mathbf{X}_{t-1},...\right)
 \end{equation}
 \end{definition}
 Given an updating rule, we can define a dynamic equilibrium, and the notion of a steady state. 
 
 \begin{definition}[Equilibrium]\label{def:steady_state} Given  a starting predicate size $S_0$, level of demand $\lambda$ and nominal block rewards $Pbr$, an equilibrium is a sequence of predicate sizes $\left\{S_t\right\}_{t={1,2,...}}$  such that at each $t$, (i)  transaction fees posted by users satisfy \eqref{fee_revenue}, (ii) the equilibrium number of miners $N_t$ satisfies free entry in \eqref{free_entry} and (iii) the law of motion for $S$ is given by \eqref{lom_gen}. A steady state is reached when changes in predicate size $S$ converge to 0.   
 \end{definition}
 
    So far, we have provided a  framework to analyze how the optimal behavior of miners and users determines fees and the energy usage of the system. In the following, we utilize this framework to analyze whether chosen laws of motion for throughput can replicate equilibrium outcomes of standard competitive markets where firms are able to autonomously adjust production in accordance with demand.\footnote{ \defref{steady_state} defines an equilibrium taking both block rewards $Pbr$ and demand $\lambda$ as given. The underlying assumption is therefore that adjustments in predicate size $S$ occur sufficiently frequently - or equivalently that the  length of periods $dt$ is sufficiently small - for the system to reach steady states between changes in demand or block rewards.}

 \section{Stablizing Congestion} 
 Unlike  standard markets, the level of supply cannot be directly chosen  by  miners, but is set by the protocol through the predicate size $S$. Absent changes in $S$, increases in demand raise congestion in the system causing exorbitant cost at little benefit to users. In this section we show that  suitable updates of throughput help overcome this limitation. Specifically, as equilibrium difficulty encodes changes in demand, we  show that updates of predicate size $S$ in accordance with changes in  difficulty allow to adjust throughput to changes in demand. This section outlines why stable levels of congestion are desirable, and shows by way of simple examples how our framework allows us to analyze the viability of various update rules. 
 
  \subsection{Characterizing Efficient Outcomes}
 Here, we briefly outline why stable levels of congestion are desirable from a welfare perspective. In equilibrium, the net benefit that the system produces can be characterized in terms of value that transactions generate across all consumers, wait times as well as the cost of energy used by pools:
  
  $$\nu\lambda-\text{Average Wait-time}\left(\lambda,S\right)-\text{Total Energy Cost}\left(S\right).$$ 
  
  Free entry implies that the total energy cost of the system must equal mining rewards. Abstracting from block rewards, total mining revenues increase at least proportionally to system congestion, as do average wait times. To ensure that the net benefits of the system do not decrease as demand $\lambda$ increases, changes in demand $\lambda$ should be offset by proportional changes in $S$. It is easy to see that this would imply that the benefit of the system would remain invariant to how many people use it. 
  
  In conclusion, an ideal adjustment of predicate size $S$ to changes in difficulty  keeps wait times  and congestion $\rho=\lambda/(\mu S)$ constant, allowing for proportional growth in demand and supply.

 \subsection{Floating Difficulty Level}
  Consider a policy rule that upon observing a change in difficulty between perdiods $t-1$ and $t$, updates predicate size $S$ in $t+1$ as follows: 
 \begin{equation} \label{eq:lom_1}
 \log\left(S_{t+1}/S_t\right)=\gamma\left(d_{t}-d_{t-1}\right).
 \end{equation}

We focus on analyzing transitions between equilibrium steady states in response to changes in demand when block rewards are zero. We analyze constraints on the updating parameter $\gamma$ stemming from two objectives. First, $\gamma$ should ensure convergence of the system to a new steady state. Second, economic efficiency requires that changes in demand yield minimal changes in congestion.

 Assume the system is in steady state with predicate size $S_0$ and congestion $\rho^*=\lambda_0/\mu S_0$. We  consider a change in demand $d\log\lambda\equiv \log(\lambda_1)-\log(\lambda_0)$ from $\lambda_0$ to $\lambda_1$. The object of interest is the elasticity of steady state system congestion with respect to this demand shock, $\frac{d\log\rho^*}{
 d\log \lambda}$.  The following proposition summarizes the properties of this object.
 \begin{proposition}\label{prop:floating}
 Consider an initial steady state $\{S_0,\lambda_0\}$ and assume that block rewards are equal to zero. Under a floating difficulty regime, upon a change in demand $d\log \lambda\equiv \log\lambda_1-\lambda_0$,
 \begin{enumerate}
 \item   the system converges to a new steady state if $\gamma<\frac{1}{\overline{\varepsilon}\left(\overline{\rho}\right)}$, where $\overline{\varepsilon}\left(\rho\right)$ denotes the upper bound on $\varepsilon\left(\rho\right)$ derived in theorem 2 and $\overline{\rho}\equiv\max\left\{\lambda_0/\mu S_0,\lambda_1/\mu S_0\right\}$. 
 
 \item if $\gamma\in \left(0,\frac{1}{\overline{\varepsilon}\left(\overline{\rho}\right)}\right)$, then steady state changes in system congestion are strictly less than proportional to changes in demand: $\frac{d\log\rho^*}{d\log\lambda}\in \left(\frac{1}{2},1\right)
 $
  \end{enumerate}
 \end{proposition}
 \begin{proof}
 Using the recursive structure of the system, write the elasticity of steady state congestion with respect to a demand shock as:
 \begin{equation*}
\frac{d\log\rho^*}{d\log\lambda} = 1+\sum_{t=1}^\infty \left(-\gamma\right)^t\left(\prod_{k=0}^t\varepsilon\left(\rho_t\right)\right)
 \end{equation*}
 Using theorem 2, the proposition follows from standard properties of geometric sums. 
 \end{proof}
 
 The first part of \propref{floating} shows that dynamic stability imposes a feasibility restriction on the rate at which predicate size (and thus throughput) can be increased in response to some demand shock. For $\gamma > 1/\overline{\varepsilon}\left(\overline{\rho}\right)$, changes in predicate size would be large enough at every step to ensure that the system diverges from equilibrium. Thus, the value of $\gamma$ needs to be low enough so that this is prevented. Note that this imposes a new trade-off in efficiency: a lower $\gamma$ ensures that higher demand shock levels can be `absorbed' without instability, but also means that the system will take longer to arrive at equilibrium (as predicate updates are less sensitive to changes in difficulty).
 
 The second part of the above result demonstrates that such a record keeping system with $\gamma > 0$ does strictly better than one where $S$ is fixed. This is due to the upper bound on the elasticity of steady state system congestion, which shows that congestion will increase at a fraction of the non-adaptive case. Although they will increase comparatively less, it is immediate that transaction fees remain unbounded: an arbitrary sequence of positive demand shocks will always increase difficulty arbitrarily high, pushing up equilibrium fees as well. Further, to ensure dynamic stability at possibly high levels of congestion requires  sufficiently low $\gamma$, which impedes the speed of convergence. In fact, the only choice for $\gamma$ that ensures dynamic stability under a floating difficulty regime as $\lambda\rightarrow\infty$ is $\gamma=0$.  
 
 In the traditional Bitcoin protocol, an unbounded difficulty parameter is required due to the security benefits that it provides. This is by making the cost of disrupting the network scale with $d$. Note, however, that an increase in difficulty here is not necessarily required for system security: the costs of participation are increasing with respect to both difficulty \textit{and} predicate size. This means that `difficulty' in the sense of Bitcoin is actually better understood to be $Sd$ here, as both of these parameters are monotonically increasing with respect to costs per unit time. This motivates the question of whether \textit{constant} equilibrium transaction fees can be achieved without compromising system security.

\subsection{Pegged Difficulty Level}
Consider a policy rule that updates predicate size $S$ as follows: 
\begin{equation}\label{eq:peg}
\log\frac{S_{t+1}}{S_t}=\gamma\left(d_t-d^*\right),
\end{equation}
where $d^*$ is a \textit{pre-specified} level of difficulty. 
  
  Again, consider a demand shock $d\log\lambda$. The following proposition draws out conditions on the update parameter $\gamma$ under which the system maintains a stable steady state level of congestion $\rho$ under zero block rewards. 
  \begin{proposition}
  Assume that block rewards are equal to zero, and difficulty is pegged to $d^*$. Denote $\bar{\rho}^*\equiv \psi^{-1}\left(\exp(d^*/\alpha)\right)$ with $\alpha$ given in $\eqref{difficulty_1}$.  Under a pegged difficulty regime, upon a change in demand $d\log \lambda\equiv \log\lambda_1-\lambda_0$ yielding a change from $\rho^*$ to $\rho$,  $\frac{d\log \rho^*}{d\log \lambda}=0$ for any $\gamma\in\left(0,\frac{2}{\overline{\varepsilon}\left(\overline{\rho}\right)}\right),$ where $\overline{\rho}\equiv \max\{\rho^*,\rho\}$. 
  \end{proposition}
  \begin{proof}
 Under a pegged difficulty regime, the equilibrium elasticity of congestion with respect to a demand shock can be written as: 
  $$\frac{d\log\rho^*}{d\log\lambda}=\varepsilon\left(\rho_0\right)\prod_{1=1}^\infty \left(1-\gamma\varepsilon\left(\rho_t\right)\right).$$
  The result follows from theorem 2 and the above expression.
  \end{proof}  
  
  Rule \ref{eq:peg} allows for the \textit{a priori} selection of some acceptable congestion level $\rho^*$ by fixing (or `hardcoding') the corresponding $d^*$ in the update rule. In this case, the system will increase predicate size $S$ accordingly until system difficulty \textit{drops} back to $d^*$. Therefore, the long-term equilibrium effect on congestion is zero. Such a system has the desirable property that transaction fees at equilibrium will always revert to constant: the effects of any demand shock will be fully compensated by changes in predicate size, which will revert transaction fees to their previous equilibrium level.
  
  Technically, the crucial change that permits such desirable behavior lies in the `decoupling' between the sufficient statistics for system security and demand in the above update rule. Indeed, the costs of participation rise monotonically with $S$, while costs of transactions are only functions of $d$. Thus, by ensuring that the system will always update $S$ as a reaction to changes in demand so as to fully compensate for any changes in $d$, we are able to keep equilibrium transaction fees constant \textit{regardless} of the demand shock as long as $\gamma$ is low enough.  
  
  \subsection{Positive Block Rewards}
  We conclude our analysis by briefly discussing how positive  block rewards, $br>0$, affect our key results. In the presence of positive nominal block rewards, miners' profit incentives are not fully determined by mining fees and therefore entry and changes in difficulty only are an imperfect measure of demand. Consequently, changes in demand cannot be fully absorbed through proportional scaling of throughput. Further, changes in nominal block rewards - either due to an increase in $P$ or a (scheduled) decrease in block rewards $br$ - will cause changes in throughput that are (potentially) detached from fundamental user demand. While nominal block rewards serve as a subsidy to sustain mining incentives at low levels of initial demand, they therefore pose a cost in terms of economic efficiency. 
  
  \section{Towards an Instantiation}\label{sec:feasibility}
  
  The feasibility of the construction in the preceding sections depends heavily on the properties of the underlying proof system utilized, as well as on the fulfillment of technical requirements for the underlying PoW scheme. In this section, we (1) outline the engineering ingredients essential to the efficient realization of such a system, (2) illustrate that the economic assumptions made are consistent with empirical observations, and (3) identify potential design challenges.
  
  \subsection{Proof System Choice} Research into the design and implementation of proof systems relevant for use in the above context has seen rapid advancements in recent years. Current implementations of Succinct Non-interactive ARguments of Knowledge (SNARKs) (see \cite{groth2016size},  \cite{ben2013snarks}, \cite{cryptoeprint:2013:279} for definitions) provide extremely fast proof verification times and small proof sizes, seeing active use in blockchain protocols (see \cite{sasson2014zerocash} for an example). Of such systems, the main relevant properties are:
  \begin{enumerate}
      \item Succinctness: Proof size needs to be small enough to add to the blockchain
      
      \item Trustlessness: The proof system imposes no trust assumptions
      
      \item Variability: Easy to produce/verify many different predicates
      
            \item Recursion: Current proofs can efficiently verify previous proofs

  \end{enumerate}
   
   Such restrictions have been the focus of active research in proof system design. Proof systems with the property of succinctness, or proof sizes that don't vary with respect to predicate size, generate proofs that can be efficiently computed and are small enough (< 1 kB) in state-of-the-art implementations such as \cite{groth2016size}. However, the majority of such systems are either reliant on cryptographic assumptions that are known to be falsifiable in the quantum setting, or come with changes in the trust model such as the need for an initial setup by a trusted third-party, thus violating notions of trustlessness. Moreover, such implementations do not provide the flexibility needed by the above model to adaptively change the number of transactions being verified in each block - doing so here comes with the corresponding prohibitive cost of rerunning the trusted setup procedure each time we update the size of $S$. This is because in these constructions the creation/validation of proofs for a given predicate require the initial generation of a predicate-specific Single Reference String (SRS) that is required in all subsequent computation.
   
   Recent work has looked into resolving the above challenges while maintaining efficiency. SNARK constructions that are based on a \textit{universal} SRS have been developed (\cite{cryptoeprint:2019:099}, \cite{cryptoeprint:2018:280}) in which the SRS generation procedure needs to be done only once for all potential predicates. This means that one SRS suffices to generate proofs for any number of transactions (i.e. any predicate size). Other designs with different trade-offs have also been developed (for more information, see \cite{cryptoeprint:2019:1076}, \cite{cryptoeprint:2019:1021}, \cite{cryptoeprint:2018:046}, \cite{cryptoeprint:2019:1047}, \cite{cryptoeprint:2019:953}).
   
   Indeed, the `ideal' proof system should be (1) trustless (rely on no trusted setup procedure), (2) plausibly quantum-resistant, (3) support proofs for multiple predicates through a universal SRS, and (4) support recursive proof composition - or the ability to verify previous proofs efficiently. The first three technological requirements together would be required for an effective instantiation of the above protocol, while (4) would cut the bandwidth costs needed by new end users to enter the system to negligible. 
   
   \subsection{Integrating Proof of Work}
   
   The core of the proposed protocol relies on a PoW procedure that generates a proof of correctness for a set number of transactions that are verified (the 'predicate size') in each block. For the proof generation procedure to act as a suitable PoW function, certain restrictive properties need to be fulfilled. This is achieved for a simple payments protocol in \cite{cryptoeprint:2020:190}, built on top of a recursive proof system. Since the ability to use proof generation as a suitable PoW puzzle is essential for the present proposal, techniques to generate PoW-suitable predicates and/or proof systems is an immediate area of interest for future work. Moreover, assessing whether the approach in \cite{cryptoeprint:2020:190} can be easily extended to proof systems based on different underlying assumptions (such as proof systems that support universal SRS generation) is a natural next question, as is assessing whether this approach can also be used in more complex predicates that support a wider array of features (such as arbitrary code execution/smart contracts). 
   
   \subsection{Distributing Proof Computation}
   
   Since the proof generation techniques that can be potentially used in the current context require substantial computational resources, in the absence of hardware-acceleration techniques this is only feasible if proof generation can be adequately distributed among workers. This approach is similar to the `mining pool' model in current PoW-based cryptocurrencies, wherein a mining pool operator (the "miner") distributes part of the PoW computation to workers who, in turn, are guaranteed a percentage of returns in exchange for working for the specific pool. Such an arrangement can provide workers with the ability to receive less income albeit more frequently, and thus has the benefit of distributing large PoW computational burdens among many small participants.
   
   In the context of proof generation, the work of \cite{wu2018dizk} provides the first framework for SNARK proof distribution among many nodes. Since the guarantees of the proposed protocol depend heavily on the proof generation cost function, empirical performance of SNARK generation in a distributed setting is crucial to assessing the accuracy of the present model. Indeed, the fundamental assumption on the functional form of the production function in this model is directly observed in the empirical results of \cite{wu2018dizk}, indicating that this model is representative of distributed proof production dynamics, at least in the context of \cite{groth2016size}. Moreover, these results suggest a potential $\sigma$ value of $\sigma \approx 0.88$. 
   
   Furthermore, distributed proving techniques for proof systems with desirable properties need to also be developed. More specifically, techniques such as those employed in \cite{wu2018dizk} should be applied to other proof systems that support additional features, such as trustlessness or universal SRS generation.  Finally, in order to fully emulate the threat model inherent to current mining pools, an effective solution needs to be devised to ensure that mining pool participants cannot submit invalid computation to the pool operators. In the present model this would require the workers to also provide a proof-of-correctness for their individual computation to the operator that is both cheap to verify (comparatively to performing the computation itself) and efficient to perform. Although this is not necessarily required in the context of trusted pools, it would be an important feature for any pool open to a public set of (potentially adversarial) workers.
   
   \subsection{Data Availability Considerations}
   
   An important issue that is relevant in the present discussion is that of data availability. Since miners are providing proofs of correctness, in order for participants to update their state accordingly the transactions that resulted in the updated blocks need to also be released. If these are added to the propagated blocks, this becomes a bottleneck in the potential throughput of such a system since block sizes cannot exceed certain sizes ($\sim 1$ MB). A promising line of work in resolving this problem uses Coded Merkle Trees \cite{cryptoeprint:2019:1139}, which can be used to reconstruct the relevant transaction data even in the context of a single honest node. Investigation of whether such an approach can be incorporated here is also a relevant avenue for future work. 
      
\section{Conclusion}
Drawing on recent developments in the literature aimed at optimizing the scaling of distributed payment systems, we have provided a protocol specification that addresses economic cost efficiencies from correctness verification, PoW security guarantees and throughput adaptability. We have shown how to link throughput to an observable system statistic that adequately reflects changes in user demand and have demonstrated by way of an economic analysis how to autonomously regulate incentives of system participants so as to ensure cost-efficient scaling to user demand. We believe that an implementation of this approach is feasible given the current state of knowledge, provided that engineering efforts yield the necessary tools that our theoretical analysis highlights. 

\newpage

\bibliographystyle{plainnat}

\appendix 
\newpage
\section{Appendix}\label{ap:proofs}

\subsection{Proof of Theorem 2}

\begin{proof}
Let $\psi\left(\rho\right)=\rho\int_{0}^{\bar{c}}\left(\bar{F}(c)-cf(c)\right)W\left(c,\rho\right)dc.$
First, the derivative of the wait time is equal to: 
\[
\frac{\partial W\left(\rho\right)}{\partial\rho}=e^{-\alpha(\rho)}\alpha(\rho)^{3}W(\rho)^{3}.
\]
 Therefore, the elasticity is given by: 
\begin{equation}\label{eq:elasticity_wait_times}
\frac{\rho}{W\left(\rho\right)}\frac{\partial W\left(\rho\right)}{\partial\rho}=\rho e^{-\alpha(\rho)}\alpha(\rho)^{3}W(\rho)^{2}.
\end{equation}

Evidently, this elasticity is increasing in $\rho.$ Thus, the elasticity
of average transaction fees is given by: 
\[
\frac{d\log\psi\left(\rho\right)}{d\log\rho}=1+\rho\int_{0}^{\bar{c}}\omega_{c}\bar{F}(c)e^{-\alpha(\rho\bar{F}(c))}\alpha(\rho\bar{F}(c))^{3}W(\rho\bar{F}(c))^{2}dc,
\]
 where $\omega_{c}\equiv\frac{\bar{F}(c)-cf(c)W\left(c,\rho\right)}{\int_{0}^{\bar{c}}\left(\bar{F}(c)-cf(c)\right)W\left(c,\rho\right)dc}.$
Given that $\alpha(\rho)\rightarrow\infty$ as $\rho\rightarrow0,$
a uniform bound on this elasticity is given by: 
\[
\frac{d\log\psi\left(\rho\right)}{d\log\rho}\leq1+\rho e^{-\alpha(\rho)}\alpha(\rho)^{3}W(\rho)^{2}\equiv \overline{\varepsilon}\left(\rho\right).
\]

\end{proof}

\end{document}